\documentclass{article} 
\usepackage{hyperref}
\usepackage{url}
\usepackage{amsthm,mathrsfs, amsmath, amsfonts,multirow}
\usepackage{graphicx}
\usepackage{algorithm,algorithmic,natbib}

\title{Calibration for Stratified Classification Models}

\author{
Chandler Zuo \\
Data Science, Audible Inc.\\
33 Washington Street, Newark, NJ 07102 \\
\texttt{chandlez@audible.com} \\
}

\newtheorem{theorem}{Theorem}
\newtheorem{lemma}[theorem]{Lemma}

\begin{document}

\maketitle

\begin{abstract}
In classification problems, sampling bias between training data and testing data is critical to the ranking performance of classification scores. Such bias can be both unintentionally introduced by data collection and intentionally introduced by the algorithm, such as under-sampling or weighting techniques applied to imbalanced data. When such sampling bias exists, using the raw classification score to rank observations in the testing data can lead to suboptimal results. In this paper, I investigate the optimal calibration strategy in general settings, and develop a practical solution for one specific sampling bias case, where the sampling bias is introduced by stratified sampling. The optimal solution is developed by analytically solving the problem of optimizing the ROC curve. For practical data, I propose a ranking algorithm for general classification models with stratified data. Numerical experiments demonstrate that the proposed algorithm effectively addresses the stratified sampling bias issue. Interestingly, the proposed method shows its potential applicability in two other machine learning areas: unsupervised learning and model ensembling, which can be future research topics.
\end{abstract}

\section{Introduction}\label{sec:introduction}

Calibrating classifier outputs to get good probability estimates has received much attention in the machine learning literature (c.f. \cite{niculescu2005obtaining}, \cite{niculescu2005predicting} and \cite{zadrozny2002transforming}). These papers have proposed and/or evaluated methodologies that mitigate poor calibration of maximum margin classification algorithms, such as Naive Bayes, Gradient Boosting and SVM. The proposed methods assume that the model-based probability scores already have good rank concordance, based on some loss function, with the "true" probability, while further adjustment is applied to reduce the bias in probability estimation. One problem that has been neglected, though, is that the rank concordance with the "true" probability can be poor in the first place, and reranking may be required. This problem can arise from the sampling bias in the training data or weighting schema, and can lead to serious problems when applying the model to holdout data sets.

To illustrate this problem, consider the following setting. Suppose covariates $X\in \mathbb R^d$ follows two different distributions in the training data and testing data, respectively denoted by $F(x)$ and $\tilde F(x)$. Also assume that $p_{Y|X}$ is the "true" probability model that generates the binary response variable $Y\in\{0,1\}$. With real data, we almost always can only fit a mis-specified model, as fitting the "true" model requires both perfect variable selection and specifiying the right form of the function (\cite{begg1990consequences}, \cite{lv2014model}). A mis-specified model, parameterized by $\theta\in \Theta$, is fitted by minimizing the average loss function over the training sample, the expectation of which is:

\begin{equation}
E_{F}l(f_\theta(X),Y)=\int \sum_{y=0}^1p_{Y|X}(y|x)w(x,y)l(f_\theta(x),y)d F(x),
\end{equation}

 where $w(x,y)$ is a weight function for each point in the sample space. Because the goal is to make predictions on the testing data, a good model should achieve low values in the average loss over the testing sample, which can be written as:

\begin{equation}
E_{\tilde F} l(f_\theta(X),Y)=\int \sum_{y=0}^1p_{Y|X}(y|x)l(f_\theta(x),y)d\tilde F(x).
\end{equation}

Given that the sampling distribution can be different from the testing, and the model is mis-specified, a sufficient condition for $E_{F}l(f_\theta(X),Y)=E_{\tilde F} l(f_\theta(X),Y)$ is $w(x,y)=d\tilde F(x)/dF(x)$. The latter ratio is called the \textit{sampling bias factor} in this paper. This condition may not be always feasible. A very common scenario is when the classes are unbalanced. In this case, $w(x,y)$ can be purposely designed to be different from the sampling bias factor to compensate for data unbalanceness in order to boost overall classification accuracy. A summary of the various strategies that either implicitly or explicitly adopt a weighting schema different from the sampling bias factor is included in \cite{ramyachitra2014imbalanced}. Given that these strategies have been demonstrated to well improve classification for unbalanced data, it is preferable to develop methods that can adjust for the sampling distribution effect post model fitting.

Optimal calibration for arbitrary sampling bias is impossible in practice, as I will show in Section \ref{sec:perfectcalib}. In this paper, I investigate calibration in a special case, where training and testing samples are stratified. In this case, optimal calibration exists if ranking within strata is kept unchanged and only ranking between strata is caliberated. Notice that building models based on stratified sampling is quite common practice for real life data. For example, in customer scoring models or credit risk models, modeling for the full data set can be challenged by data volume, and a common practice is to sample data stratified on pre-defined customer segments for initial model training. Calibrating ranking across the sampling strata is necessary according to the finding in this paper, but to my best knowledge, because such problem has not been emphasized in literature, ranking calibration is often neglected by practitioners, which can leads to sub-optimal or even inaccurate results.

In this paper, I propose a calibration method that ad-hocly re-ranks the classification scores to address the training data sampling bias. The method is developed by converting the calibration problem to an optimization problem for the ROC curve conditioning on the testing data distribution. Mathematical solution to this problem shows that the optimal ranking should be based on quantities that are closely connected to the slope along the ROC curve. Earlier research (\cite{choi1998slopes} and \cite{johnson2004advantages}) has established the relationship between slopes along the ROC curve with likelihood ratios, which lends to theoretical interpretation of my newly proposed methodology. 

This paper is organized as the following. In Section \ref{sec:rerank} I derive mathematically the optimal re-ranking solution to address sampling bias. The solution requires assumptions that both sampling bias function and conditional class probabilities are known. In Section \ref{sec:application}, I discuss how the mathematical results can be applied in three different real life settings. Section \ref{sec:examples} includes numerical studies to demonstrate the performance of the proposed methods, and Section \ref{sec:conclusions} discusses the potential extentiation of this research onto future areas.

\section{Ranking Calibration to Optimize ROC Curve}\label{sec:rerank}

The goal of my proposed re-ranking method is to optimize the ROC curve. ROC is a state-of-the-art metric to measure binary classification model performances \cite{fawcett2006introduction}, and two of its properties lend itself to the purpose of our calibration goal: (1) its invariance to the balance of class labels; (2) it is only based on the ranking but not the scale of classification scores. The second property is especially relevant to this paper since my main goal is to calibrate the ranking due to sampling bias.

Consider a classification model with scoring function $\hat f: \mathcal X \rightarrow \mathbb R$. The ROC curve is generated by computing the sequence of True Positive Rates (TPR) and False Positive Rates (FPR) at varying score thresholds $\lambda$. I use $\alpha(\lambda)$ and $\beta(\lambda)$ to denote the $FPR$ and $TPR$ respectively corresponding to the threshold $\lambda$ on the testing data. Their values can be written as:

\[
\alpha(\lambda)=E_{\tilde F}(\hat f(X) > \lambda | Y = 0)=\int 1\{\hat f(x) > \lambda\}\frac{p_{Y|X}(0|x)}{p_Y(0)}d\tilde F(x),
\]

\[
\beta(\lambda)=E_{\tilde F}(\hat f(X) > \lambda | Y = 1)=\int 1\{\hat f(x) > \lambda\}\frac{p_{Y|X}(1|x)}{p_Y(1)}d\tilde F(x).
\]

In order for re-ranking, I generalize the staticly valued $\lambda$ as a function of the covariates $\lambda(x)$, which sets the score threshold at different points in the sample space: 

\[
\alpha(\lambda)=E_{\tilde F}(\hat f(X) > \lambda(X) | Y = 0)=\int 1\{\hat f(x) > \lambda(x)\}\frac{p_{Y|X}(0|x)}{p_Y(0)}d\tilde F(x),
\]

\[
\beta(\lambda)=E_{\tilde F}(\hat f(X) > \lambda(X) | Y = 1)=\int 1\{\hat f(x) > \lambda(x)\}\frac{p_{Y|X}(1|x)}{p_Y(1)}d\tilde F(x).
\]

Optimizing the ROC curve requires maximizing TPR at each FPR level. Notice that a global solution that maximizes TPR at all FPR levels may not exist in general ROC optimization problems (\cite{fawcett2006introduction}), but it exists under certain conditions. In the rest of this section, I consider ROC optimization in two situations: unconditional and conditional on stratified sampling. 

\subsection{Perfect Calibration}\label{sec:perfectcalib}

In the unconditional case, I consider $\lambda$ as an unconstrained function.  The problem to maximize TPR at one FPR level $\alpha_0$ is:

\begin{equation}\label{eq:optim}
\begin{split}
\max & \int 1\{\hat f(x) \geq \lambda(x)\} \frac{p_{Y|X}(1|x)}{p_Y(1)}d\tilde F(x),\\
s.t. & \int 1\{\hat f(x) \geq \lambda(x)\} \frac{p_{Y|X}(0|x)}{p_Y(0)}d\tilde F(x) \leq \alpha_0.
\end{split}
\end{equation}

\begin{lemma}\label{lm:rank}
Consider $R(c)=\{x:p(1|x)p_Y(0)\geq cp(0|x)p_Y(1)\}$. Assume that 
\begin{enumerate}
\item $p_{Y|X}(1|x)\in (0,1),~a.s. ~\tilde F$;
\item $\exists c_0~s.t.~\int_{R(c_0)}\frac{p(0|x)}{p_Y(0)}d\tilde F(x)=\alpha_0$.
\end{enumerate}
Then, the solution to the optimization problem \ref{eq:optim} is $\lambda=\lambda_0$ satisfying $1\{\hat f(x)\geq \lambda_0(x)\}=1(R(c_0)),~a.s.$.
\end{lemma}

\begin{proof}
Consider an arbitrary feasible solution for \ref{eq:optim}, $\lambda$. Let $R=\{x:\hat f(x)\geq \lambda(x)\}$. When $x\in R\setminus R(c_0)$, $p_{Y|X}(1|x)/p_Y(1)<c_0p_{Y|X}(0|x)/p_Y(0)$. When $x\in R(c_0)\setminus R$, $p_{Y|X}(1|x)/p_Y(1)\geq c_0p_{Y|X}(0|x)/p_Y(0)$. Therefore,

\begin{equation*}
\begin{split}
&\int_{R(c_0)}\frac{p_{Y|X}(1|x)}{p_Y(1)}d\tilde F(x)-\int_{R}\frac{p_{Y|X}(1|x)}{p_Y(1)}d\tilde F(x)\\
&=\int_{R(c_0)\setminus R}\frac{p_{Y|X}(1|x)}{p_Y(1)}d\tilde F(x)-\int_{R\setminus R(c_0)}\frac{p_{Y|X}(1|x)}{p_Y(1)}d\tilde F(x)\\
&\geq c_0[\int_{R(c_0)\setminus R}\frac{p_{Y|X}(0|x)}{p_Y(0)}d\tilde F(x)-\int_{R\setminus R(c_0)}\frac{p_{Y|X}(0|x)}{p_Y(0)}d\tilde F(x)]\\
&= c_0[\int_{R(c_0)}\frac{p_{Y|X}(0|x)}{p_Y(0)}d\tilde F(x)-\int_{R}\frac{p_{Y|X}(0|x)}{p_Y(0)}d\tilde F(x)]\\
&=c_0[\alpha_0-\int1\{\hat f(x) \geq \lambda(x)\}\frac{p_{Y|X}(0|x)}{p_Y(0)}d\tilde F(x)]\geq 0.
\end{split}
\end{equation*}

Notice that the last inequality is because $\lambda$ is a feasible solution for \ref{eq:optim}. This finishes the proof.
\end{proof}

Although Lemma \ref{lm:rank} does not provide an explicit solution for function $\lambda$, it shows that the optimal threshold is equivalent to ranking by 

\[
OR(x):=p(1|x)p_Y(0)/[p(0|x)p_Y(1)],
\]

which is the odds ratio between the conditional distribution given covariates and the marginal distribution of $Y$. With additional steps, I demonstrate two additional important properties of ranking by odds ratio: (1) such ranking uniformly optimizes the ROC curve at all FPR levels; (2) the optimality of such ranking is invariant to different sampling distributions.

\begin{theorem}\label{th:rank}
Consider $R(c)=\{x:OR(x)\geq c\}$. For any distribution $\tilde F'$ on $\mathbb R^p$ which satisfies: 
\begin{enumerate}
\item $p_{Y|X}(1|x)\in (0,1),~a.s. ~\tilde F'$;
\item $\forall \alpha_0\in (0,1),~\exists c_0~s.t.~\int_{R(c_0)}\frac{p(0|x)}{p_Y(0)}d\tilde F'(x)=\alpha_0$.
\end{enumerate}
Then, ranking data points according to $OR(x)$ uniformly maximizes TPR at all FPR levels corresponding to $\tilde F'$.
\end{theorem}

\begin{proof}
Given an FPR level $\alpha_0$, the assumption shows that $\exists c_0~s.t.~\int_{R(c_0)}\frac{p(0|x)}{p_Y(0)}d\tilde F'(x)=\alpha_0$. Following similar arguments in the proof of Lemma \ref{lm:rank}, it can be seen that TPR is maximized at FPR level $\alpha_0$ if data points in $R(c_0)$ are ranked before $R_(c_0)^C$. Increase $\alpha_0$ from 0 to 1, the corresponding $c_0$ decreases, and the set sequence $\{R(c_0)\}$ is nested since $R(c)\subset R(c')$ when $c>c'$. This shows that ranking by $OR(x)$ maximizes TPR at all FPR levels.
\end{proof}

The implication of Theorem \ref{th:rank} is two fold. First, it shows that optimizing ROC curve has the same goal as recovering the "true" model $p_{Y|X}$. This justifies the usage of ROC as the optimziation criteria, because better ROC curve shows better approximation to the "true" model. To my best knowledge, this relationship between ROC curve and "true" model is new finding in the literature. Second, such optimal ranking is impossible in practice, as either models are mis-specified or the parameters are not estimated exactly. This suggests me to explore ROC optimization under special conditions.

\subsection{Calibration for Bias Conditional on Sampling Strata}

In this section, I consider that both training data and testing data are stratified sampled based on $(G,Y)$ where $G=g(X)$. The marginal distributions for $G$ are $F_G$ and $\tilde F_G$ for training data and testing data respectively. Conditional on $G$, the probability for $Y$ is $p_{Y|G}$ and $\tilde p_{Y|G}$ respectively on the training and testing data. Conditional on $G,Y$, $(X|G,Y)$ has the same distribution between training and testing data. I further consider $\lambda$ as a function of $G$, so that it only adjusts for the sampling on the strata. The problem to maximize TPR at one FPR level $\alpha_0$ becomes:

\begin{equation}\label{eq:optim1}
\begin{split}
\max & \frac{\int P[\hat f(X) \geq \lambda(g)|Y=1,G=g]\tilde p_{Y|G}(1|g)d\tilde F_G(g)}{\int \tilde p_{Y|G}(1|g)d\tilde F_G(g)},\\
s.t.& ~\frac{\int P[\hat f(X) \geq \lambda(g)|Y=0,G=g] \tilde p_{Y|G}(0|g) d\tilde F_G(g)}{\int \tilde p_{Y|G}(0|g)d\tilde F_G(g)}\leq \alpha_0.
\end{split}
\end{equation}

\subsubsection{Optimal Stratified Calibration}

\begin{lemma}\label{lm:condrank}
Denote $\alpha(\eta, g)=P(\hat f(X)\geq \eta|Y=0,G=g)$, $\beta(\eta, g)=P(\hat f(X)\geq \eta|Y=1,G=g)$. Let
\[
h_g(\eta)=\frac{\partial\beta(\eta,g)/\partial\eta}{\partial\alpha(\eta,g)/\partial\eta}.
\]
Assume that
\begin{enumerate}
\item $p_{Y|X}(1|x)\in (0,1),~a.s. ~\tilde F_G$;
\item $\forall g~a.s.~\tilde F_g,~h_g$ is continous in $\eta$, $\lim_{\eta\rightarrow -\infty}h_g(\eta)=\infty$, $\lim_{\eta\rightarrow \infty}h_g(\eta)=0$.
\end{enumerate}
Then, 
\begin{enumerate}
\item $\forall g~a.s.~\tilde F_g$, $\forall c\in \mathbb R$, $\exists \lambda(c,g)$ s.t. $\tilde p_{Y|G}(1|g)h_g(\lambda(c,g))/\tilde p_{Y|G}(0|g)=c$.
\item The equation in $c$: $\int \tilde p_{Y|G}(0|g)\alpha(\lambda(c,g),g)d\tilde F_G(g)=\alpha_0\int \tilde p_{Y|G}(0|g)d\tilde F_G(g)$ has solution $c=c(\alpha_0)$.
\item $\lambda(c(\alpha_0),g)$ as a function of $g$ is the solution to the optimization problem \ref{eq:optim1}.
\end{enumerate}
\end{lemma}

\begin{proof}
The first statement is obvious by the assumption that $h_g$ is continuous and has the image of $\mathbb R$. For the second statement, notice that the continuity of $h_g$ implies that for fixed $g$, (1) $\lambda(c,g)$ is continuous in $c$; (2) $\alpha(\eta,g)$ is continous in $\eta$. This shows that the left integral in the second statement, as a function of $c$ is continuous. By the second assumption, this integral converges to $\infty$ as $c\rightarrow 0$ and to 0 as $c\rightarrow \infty$. Thus, $c=c(\alpha_0)$ exists.

To show the final statement, consider the Lagrange multipler for \ref{eq:optim1}:

\[
\begin{split}
La(L)=&\int P[\hat f(X) \geq \lambda(g)|Y=1,G=g]\tilde p_{Y|G}(1|g)d\tilde F_G(g)\\
&-L(\int P[\hat f(X) \geq \lambda(g)|Y=0,G=g]\tilde p_{Y|G}(0|g) d\tilde F_G(g)-\alpha_0C)
\end{split}
\]

Differentiate w.r.t. $\lambda(g)$ for each $g$, it can be seen that the solution to the problem need satisfy:

\[
\begin{split}
&h_g(\eta)\cdot \tilde p_{Y|G}(1|g)/\tilde p_{Y|G}(0|g)|_{\eta=\lambda(g)}=c,~\forall g~a.s.~\tilde F_G\\
&\frac{\int P[\hat f(X) \geq \lambda(g)|Y=0,G=g] \tilde p_{Y|G}(0|g) d\tilde F_G(g)}{\int \tilde p_{Y|G}(0|g)d\tilde F_G(g)} = \alpha_0.
\end{split}
\]

The existence of such solution is proved by the first two results. This concludes the proof of Theorem \ref{th:condrank}.
\end{proof}

From the proof above, it is easy to see that the set $\{x:\hat f(x)>\lambda(c(\alpha_0),g(x))\}$ is nested as $\alpha_0$ increases from $-\infty$ to $\infty$. Using the argument similar to Theorem \ref{th:rank}, I obtain the following statement:

\begin{theorem}\label{th:condrank}
Define the ranking function
\begin{equation}\label{eq:rx}
r(x)=h_{g(x)}(\hat f(x))\cdot \tilde p_{Y|G}(1|g(x))/\tilde p_{Y|G}(0|g(x)).
\end{equation}
Assume that the conditions in Lemma \ref{lm:condrank} holds for any $\alpha_0\in(0,1)$. Using $r(x)$ to rank all data points in the testing set uniformly optimizes TPR at all FPR levels.
\end{theorem}

Theorem \ref{th:condrank} has its value in practice. For fixed $g$, the set of point pairs $\{(\alpha(\eta,g),\beta(\eta,g)):-\infty<\eta<\infty\}$ forms the ROC curve conditional on $G=g$, and $h_g$ is the slope on this curve. Thus, the result of this theorem basically states that the optimal ranking should be based on the slope of such conditional ROC curves, adjusting for the conditional odds ratio $\tilde p_{Y|G}(1|g)/\tilde p_{Y|G}(0|g)$. For this reason, in the subsequent sections I will refer to $h_g$ by the \textit{slope function}.

Noticeably, the slope of ROC curve has drawn earlier attentions by \cite{choi1998slopes} and \cite{johnson2004advantages}, which show that it represents the likelihood ratio at a single datum, and advocate its usage for classification models. Following the same line, the ranking quantity $r(X)$ can be interpreted as the conditional likelihood ratio evaluated at the sampling distribution of the testing set. To see this, consider the slope of the conditional ROC curve at point $(x,y)$. It can be seen that the slope of the conditional ROC curve at data point $(x,y)$ is:

\[
h_g(\hat f(x))=\frac{P(\hat f(X)=\hat f(x)|Y=1,G=g)}{P(\hat f(X)=\hat f(x)|Y=0,G=g)}.
\]

Therefore, the ranking function can be rewritten as:

\[
r(x)=\frac{P(\hat f(X)=\hat f(x)|Y=1,G=g)}{P(\hat f(X)=\hat f(x)|Y=0,G=g)}\cdot\frac{\tilde p_{Y|G}(1|g)}{\tilde p_{Y|G}(0|g)}=\frac{P_{te}(\hat f(X)=\hat f(x),Y=1|G=g)}{P_{te}(\hat f(X)=\hat f(x),Y=0|G=g)},
\]

where $P_{te}$ means that the probability calculation is based on the testing data distribution. The right hand side is indeed the likelihood ratio for the pair of random variable $(\hat f(X),Y)$ conditional on $G$. This shows that ranking data by $r(x)$ is actually ranking based on the conditional joint distribution of $(\hat f(X), Y)$ among the testing sample. By calculating the probability conditional on $G$, $r(x)$ incorporates strata information to adjust for the sampling difference in the testing data. In addition, by looking at the likelihood ratio, $r(x)$ has the same scale regardless of the scale of the raw classification score $\hat f$. In practice, this means that different classification models can be applied on different sampling strata and then integrated easily by using $r(x)$.

\subsubsection{Application of Theorem \ref{th:condrank}}\label{sec:application}

In this section I discuss three potential applications for Theorem \ref{th:condrank}.

\textbf{Stratified Sampling within Training Data}

Assume that the full training data have the same distribution as testing data, but classification models are trained on stratified sampled training data. In this case, $\tilde p_{Y|G}(1|g)$ can be estimated from the full training data, and $h_g(\hat f(x))$ can be estimated from the conditional ROC curve. The details in ranking the testing data are described by Algorithm \ref{alg:rank}.

\begin{algorithm}
\caption{Ranking for Stratified Sampled Training Data}
\label{alg:rank}
\begin{algorithmic}[1]
\REQUIRE A labeled training data set with $J$ groups, denoted by $\cup_{j=1}^J \mathcal D_j$, where $\mathcal D_j=\{(x_{i,j},y_{i,j}):1\leq i\leq n_i\}$, and an unlabeled testing data set with $J$ groups, denoted by $\cup_{j=1}^J \mathcal D'_j$, where $\mathcal D'_j=\{x'_{i,j}: 1\leq i\leq n'_i,1\leq j\leq J\}$.
\FOR {j=1 to J}
\STATE Fit a classification model using data $\{(x_{i,j},y_{i,j}):1\leq i\leq n_i\}$, with sampling or weighting conditional on $y_{i,j}$. Denote the classification score function by $\hat f_j$.
\STATE Compute the TPR and FPR at each point in $\mathcal D_j$. Estimate the interpolated function $\hat \alpha(s)$ at an arbitrary score value $s$.
\STATE Estimate a function $\hat h_j$ for the slope of the ROC curve at each FPR level $\alpha$ using smoothing methods. A review of such methods are included in \cite{gonccalves2014roc}.
\STATE Estimate the conditional odds at this strata by $\hat o_j=(\sum_{i=1}^{n_i}y_{ij})/[\sum_{i=1}^{n_i}(1-y_{ij})]$.
\STATE Compute the ranking metric for each sample point in $D'_j$ by $\hat r(x'_{i,j})=\hat o_j\cdot \hat h_j(\hat\alpha(\hat f(x'_{i,j})))$.
\ENDFOR
\STATE Rank data in the testing data according to the value $\hat r(x'_{i,j})$.
\end{algorithmic}
\end{algorithm}

\textbf{Semi-supervised Learning}\label{sec:semi}

When the distributions between training and testing data are diffrerent, applying a mis-specified model on the testing data can be sub-optimal. One remedy suggested by Theorem \ref{th:condrank} is to find a strata variable $G$ such that (1) distribution in $G$ is different between training and testing data; (2) at different strata $P(Y=1|G=g)$ are different. While finding a strata variable that satisfy both conditions can be difficult, the first condition is relatively easy to obtain via unsupervised learning. One idea is to pool the training and testing data set and apply the Principle Component Analysis, and use the first principle component as the strata variable. Exploraty analysis is recommended to see whether training data and testing data are distributed differently in the direction of the first principle component, and whether the conditional probability $P(Y=1|PC1)$ has sufficient variability. If such initial exploration suggests that the first component is indeed a good candidate strata variable, Algorithm \ref{alg:pca} can be applied to stratify the data, before Algorithm \ref{alg:rank} is applied to calculate the actual ranking scores.

\begin{algorithm}
\caption{Determining the Strata Variable}
\label{alg:pca}
\begin{algorithmic}[1]
\REQUIRE A labeled training data set $\{(x_{i},y_{i}):1\leq i\leq n\}$, and an unlabeled testing data set $\{x'_{i}:1\leq i\leq n'\}$.
\STATE Train a classification score function $\hat f$ based on the training data.
\STATE Perform Principle Component Analysis based on the pooled data set $\{x_{i}\}\cup\{x'_{i}\}$. Calculate the first principle component score $s(x)$ for each sample point.
\STATE Determine the number of strata $J$ as well as sequence of thresholds $-\infty < c_1 <\cdots<c_J<c_{J+1}=\infty$. Let $\mathcal D_j=\{(x_{i},y_{i}):c_j\leq s(x_i)<c_{j+1}\}$ and $\mathcal D'_j=\{x'_{i}:c_j\leq s(x'_i)<c_{j+1}\}$.
\end{algorithmic}
\end{algorithm}

\textbf{Model Ensembling}\label{eq:ensemble}

Consider the extremely scenario where $g(x)=x$ in Theorem \ref{th:condrank}. In this case, Eqn. \ref{eq:rx} becomes:

\[
r(x)=h_x(\hat f(x))\tilde p_{Y|X}(1|x)/\tilde p_{Y|X}(0|x).
\]

This equation suggests one way to ensemble two models: one model that produces the classification score function $\hat f$ used to estimate the slope of the ROC curve, the other that produces the probability estimate for $\tilde p_{Y|X}$ used for calibration. In practice, this suggests a new way to ensemble models: the formula can be applied iteratively to ensemble a large number of models. Computationally, estimating $h_x$ requires nonparametric smoothing, making this formula hard to compete with common ensembling techniques such as averaging or polling. Nevertheless, the serendipity finding for the relationship between ensembling and calibration is interesting, and this may lead to new model ensembling techniques.

\subsubsection{Alternative Methods}

When training and testing data are stratified sampled, and the distribution across strata are different, two alternative methods are frequently applied in practice. The first method is to calibrate the classification score within individual strata to obtain probability estimates, and use such estimates to rank data points across strata. Theoretically, this method performs best when the calibrated probability estimates is the "true" model $p_{Y|X}$; practically, because the model is always mis-specified, performance of probability estimates can be sub-optimal. The second method is to embed the sampling bias factor as weights in the loss function on the training data. As I mentioned in Section \ref{sec:introduction}, such a method is not always applicable since certain classification models do not permit such weighting schema. Comparison between these methods and my proposed method is further investigated in Section \ref{sec:examples}.

\section{Examples}\label{sec:examples}

In this section, I use the credit card data set from \cite{dal2015calibrating} to illustrate the application of the theory in Section \ref{sec:rerank}. This data set contains transactions made by credit cards in September 2013 by European cardholders. The 284,807 transactions presented in this data set occurred in two days, among which 492 were frauds. The data set includes 28 transformed features as predictors, one binary fraud indicator, one variable for the fraudulent amount, and a timestamp. In my experiments, I disregard the fraudulent amount, and only consider classifying fraud v.s. non-fraud transactions using the 28 predictors. Based on the timestamp variable, I divide the data into a training set, a tuning set and a testing set, each representing 49\%, 21\% and 30\% of the whole time period covered in the data set. The area under the ROC curve for the testing set is the evaluation criteria for all different methods.

\subsection{Calibration with Stratified Sampling}

\begin{table}[htbp]
\centering
\caption{Average AUROC across different models and unbalancing techniques.}\label{tbl:stratasample}
\begin{tabular}{ccccccc}
\hline
& \multicolumn{2}{c}{Weighting} & \multicolumn{2}{c}{Under-sampling} & \multicolumn{2}{c}{SMOTE}\\
Calibration?& Yes & No & Yes & No & Yes & No\\\hline
xgboost & 0.898 & 0.863 & 0.871 & 0.845 & 0.904 & 0.862\\
SVM & 0.925 & 0.894 & 0.925 & 0.83 & 0.933 & 0.887\\
glmnet & 0.892 & 0.804 & 0.869 & 0.772 & 0.884 & 0.789\\\hline
\end{tabular}
\end{table}

\begin{figure}
\centering
\begin{tabular}{ccc}
\includegraphics[width=0.33\textwidth]{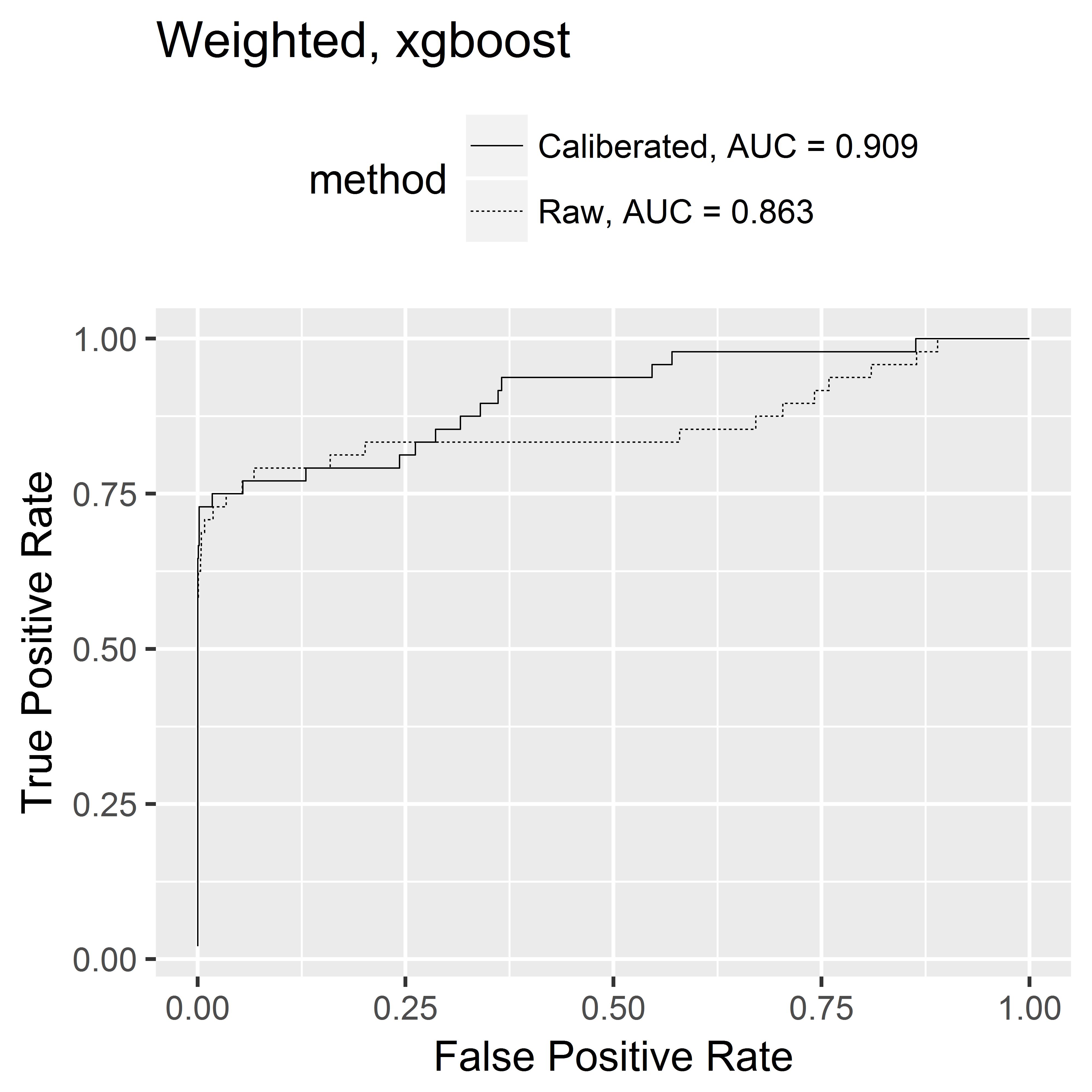} & 
\includegraphics[width=0.33\textwidth]{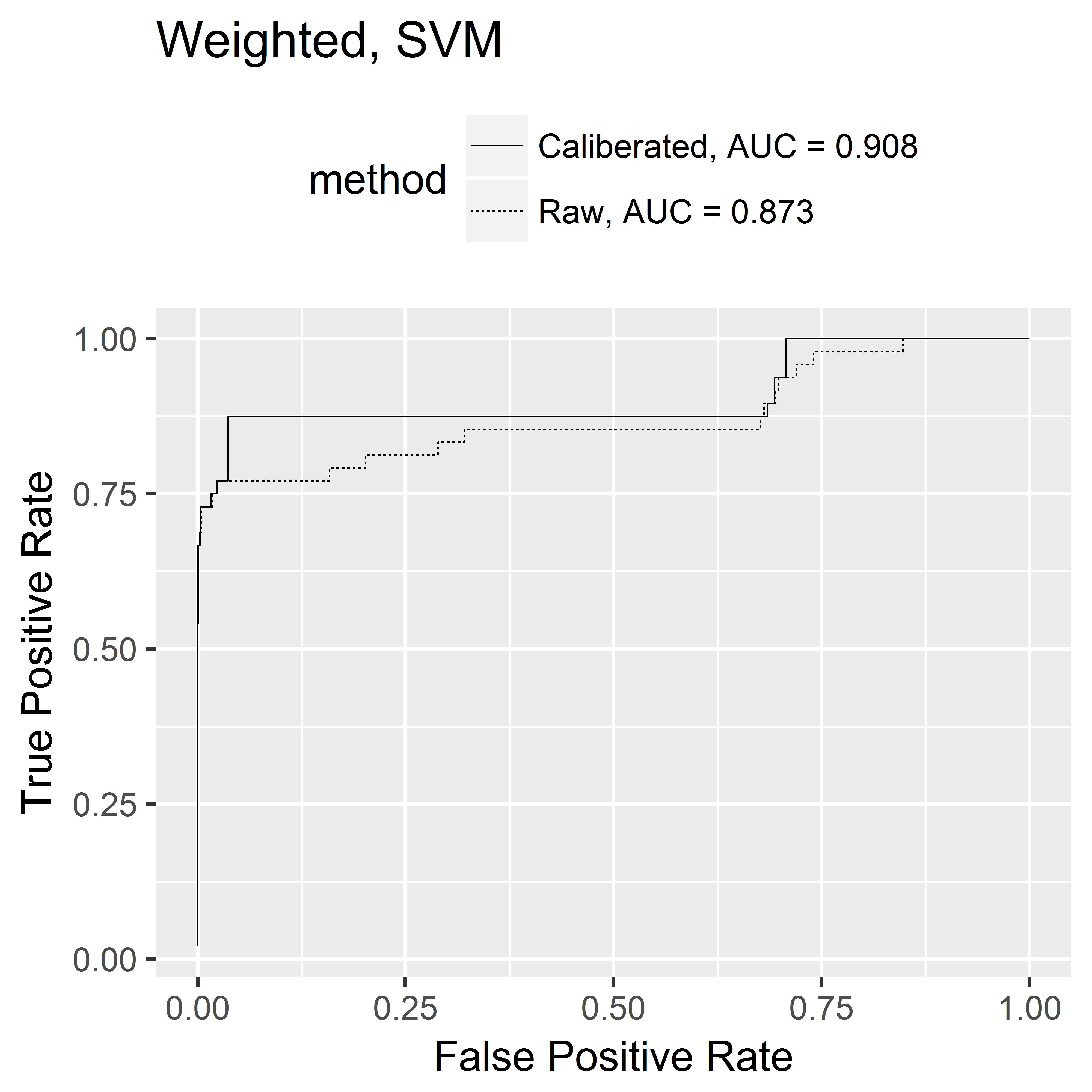} &
\includegraphics[width=0.33\textwidth]{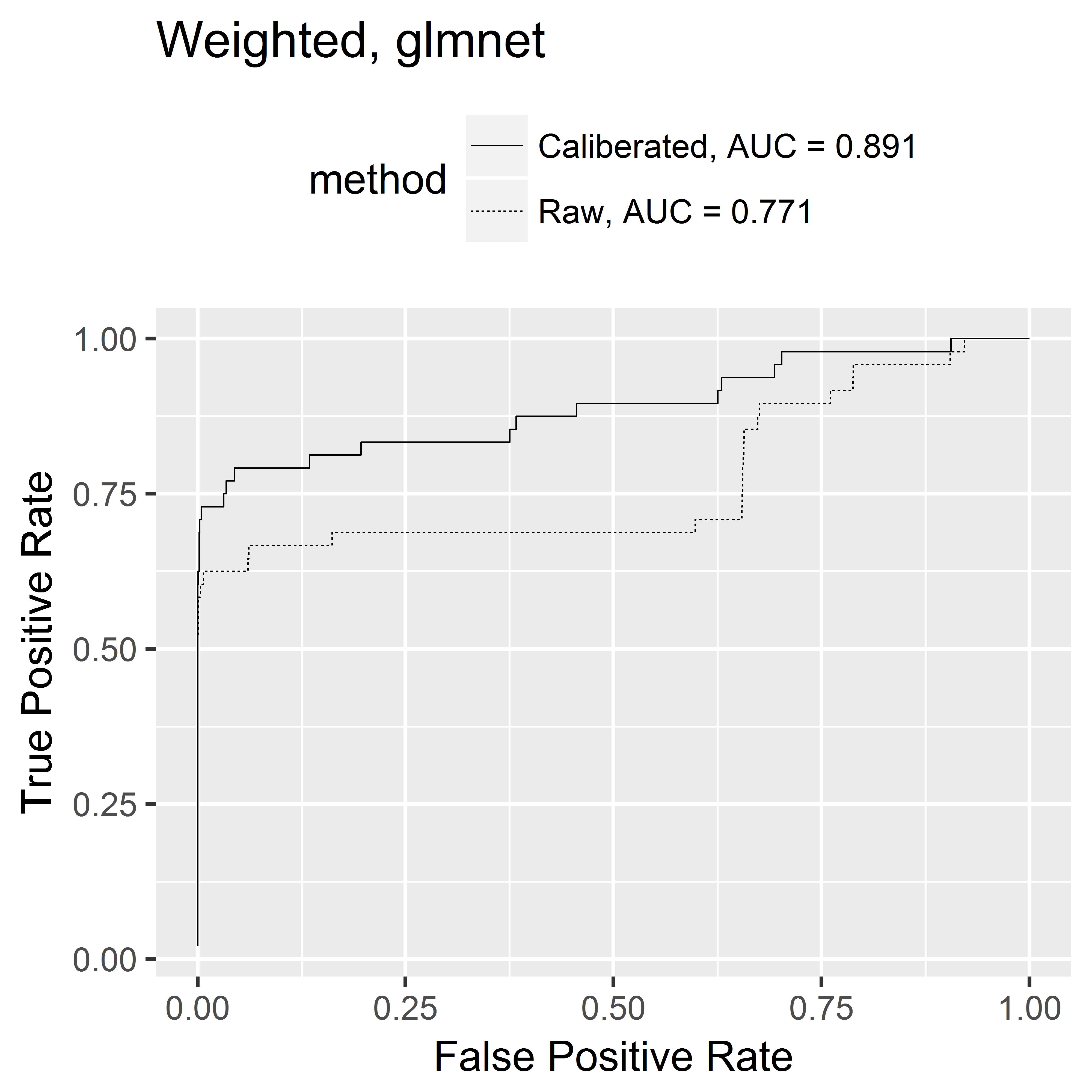} \\
\includegraphics[width=0.33\textwidth]{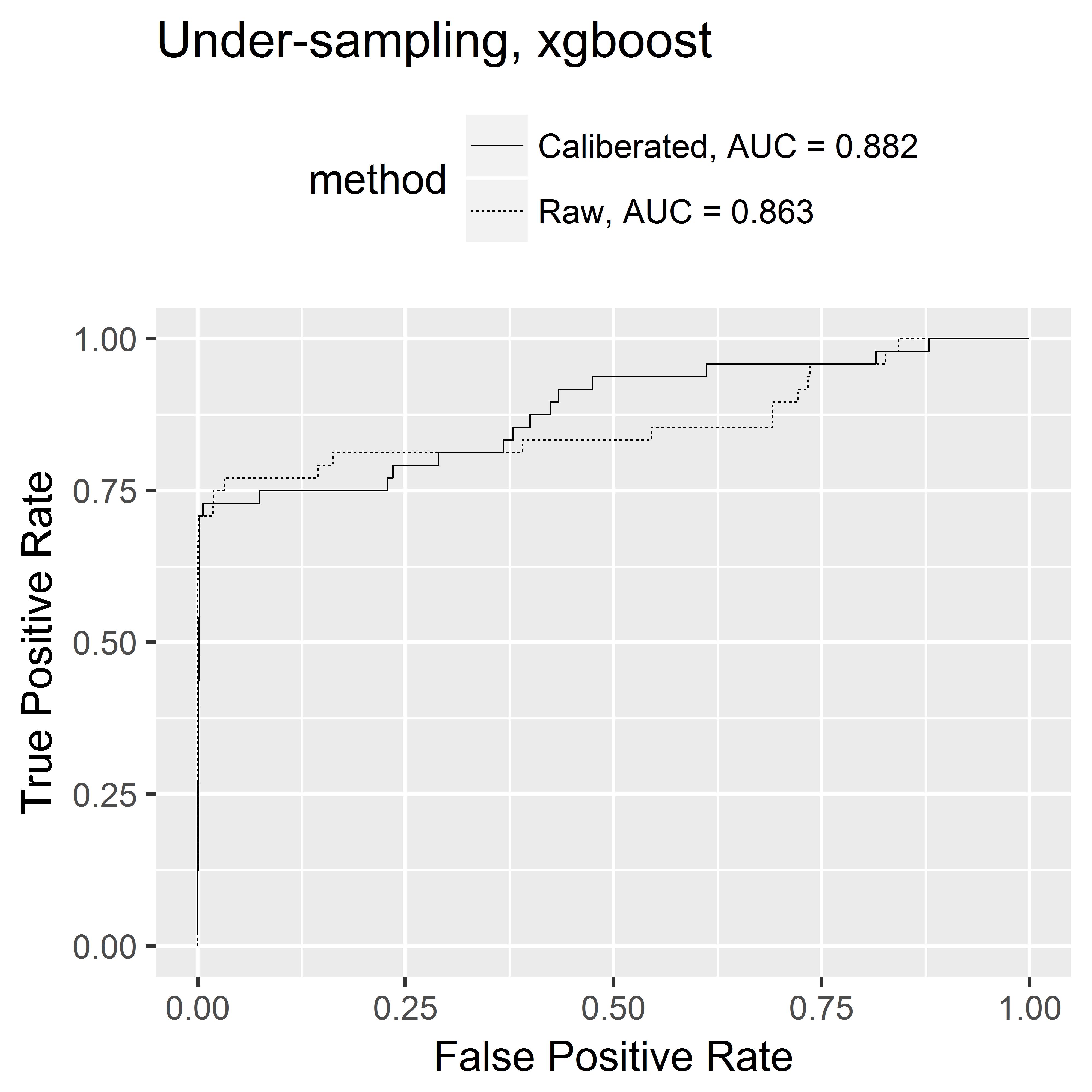} & 
\includegraphics[width=0.33\textwidth]{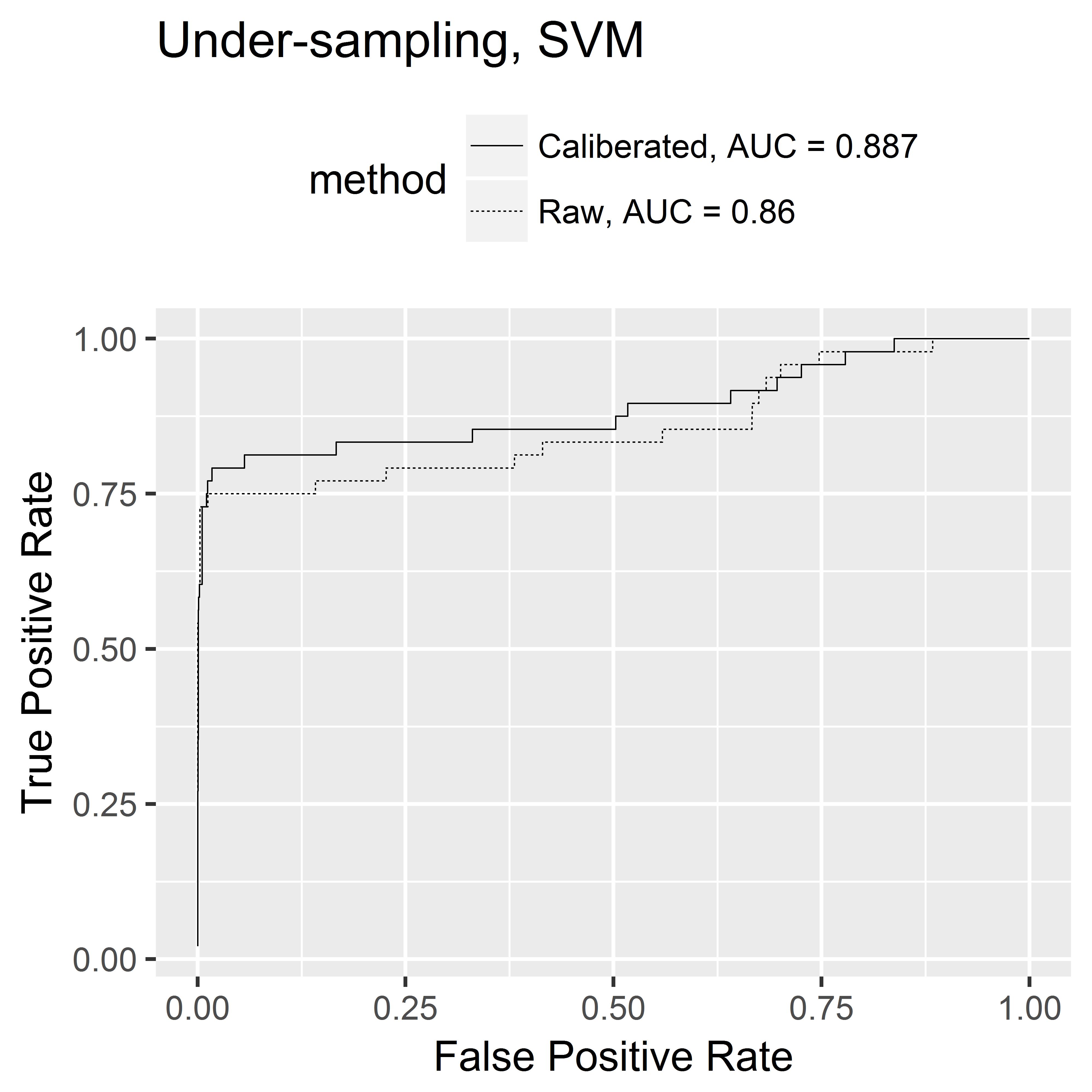} &
\includegraphics[width=0.33\textwidth]{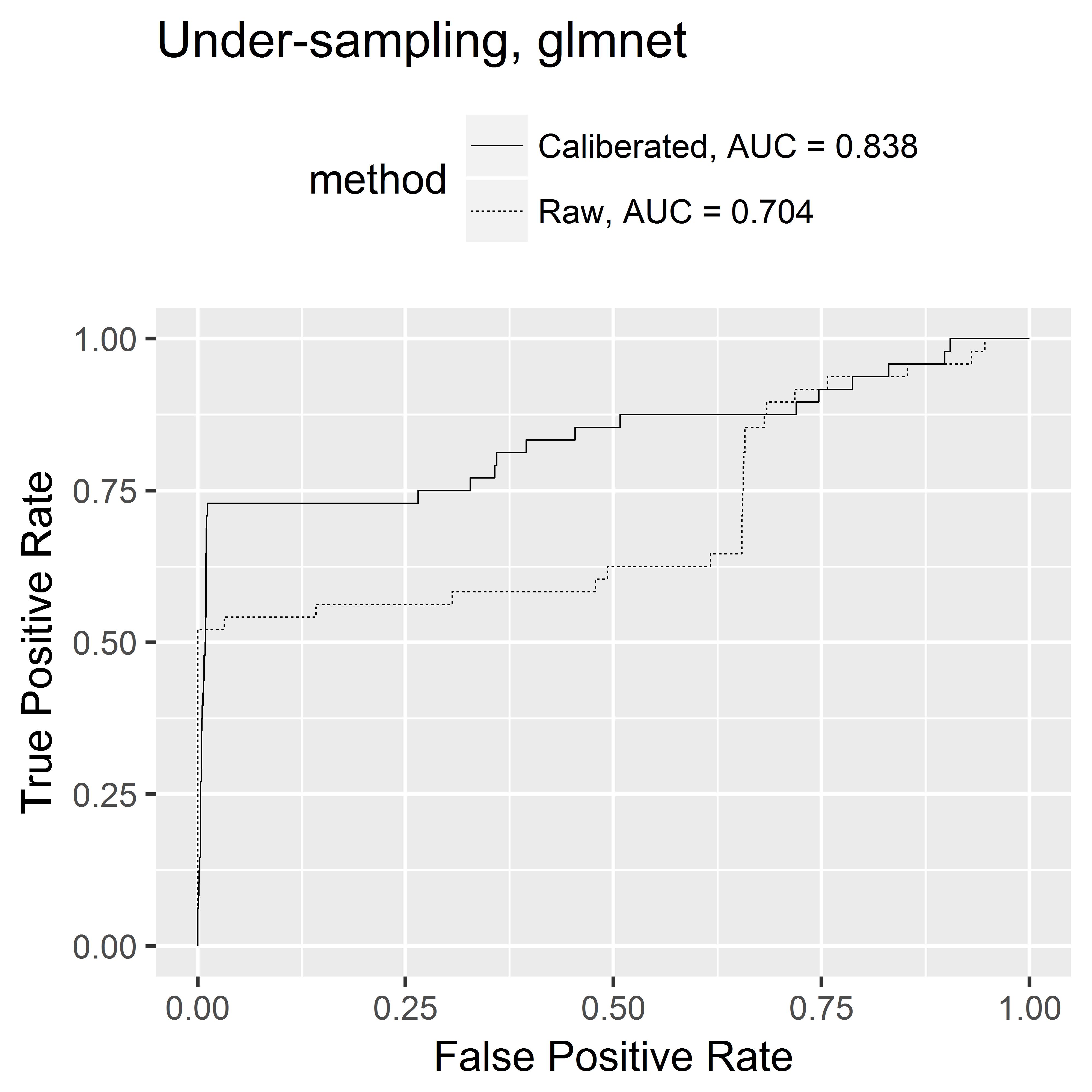} \\
\includegraphics[width=0.33\textwidth]{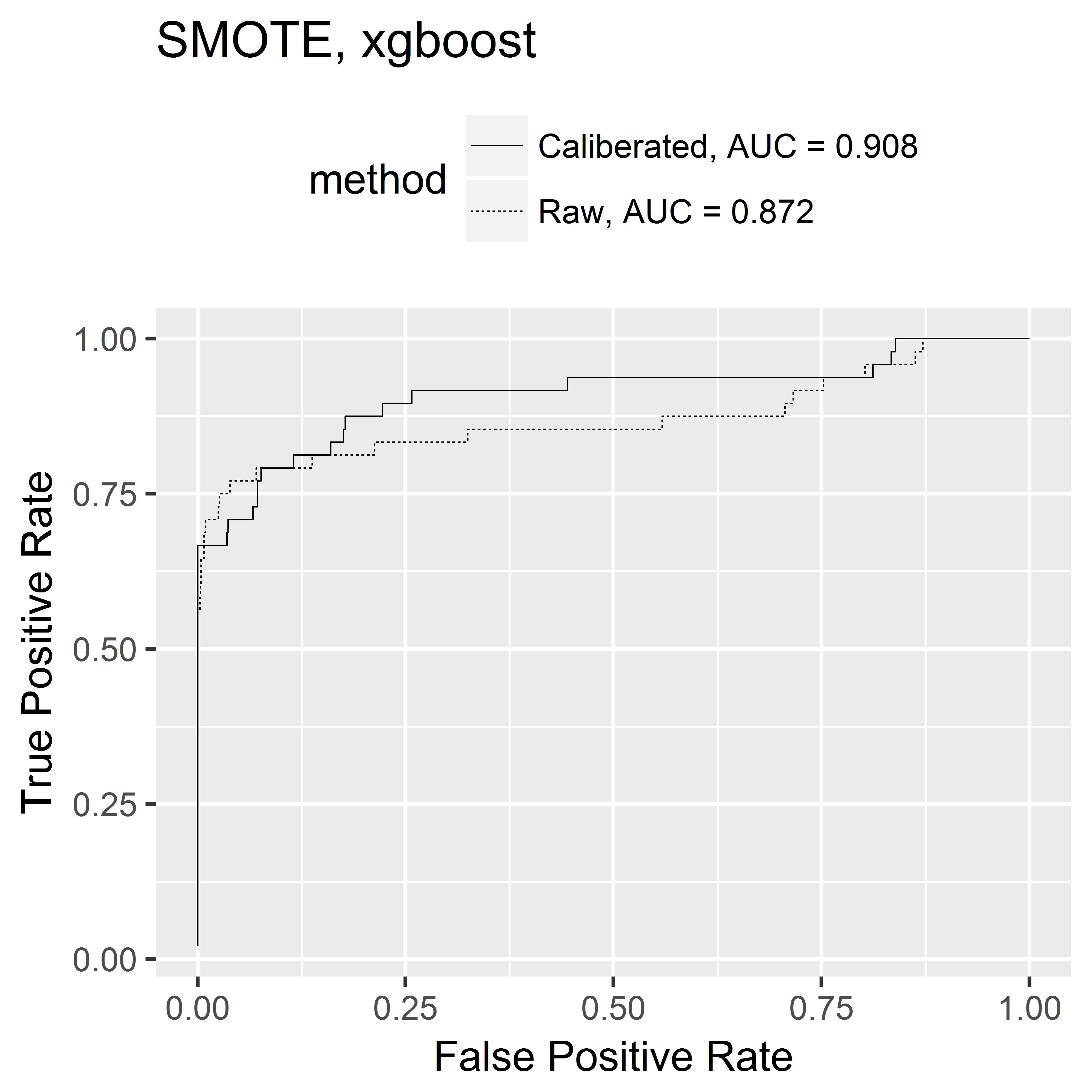} & 
\includegraphics[width=0.33\textwidth]{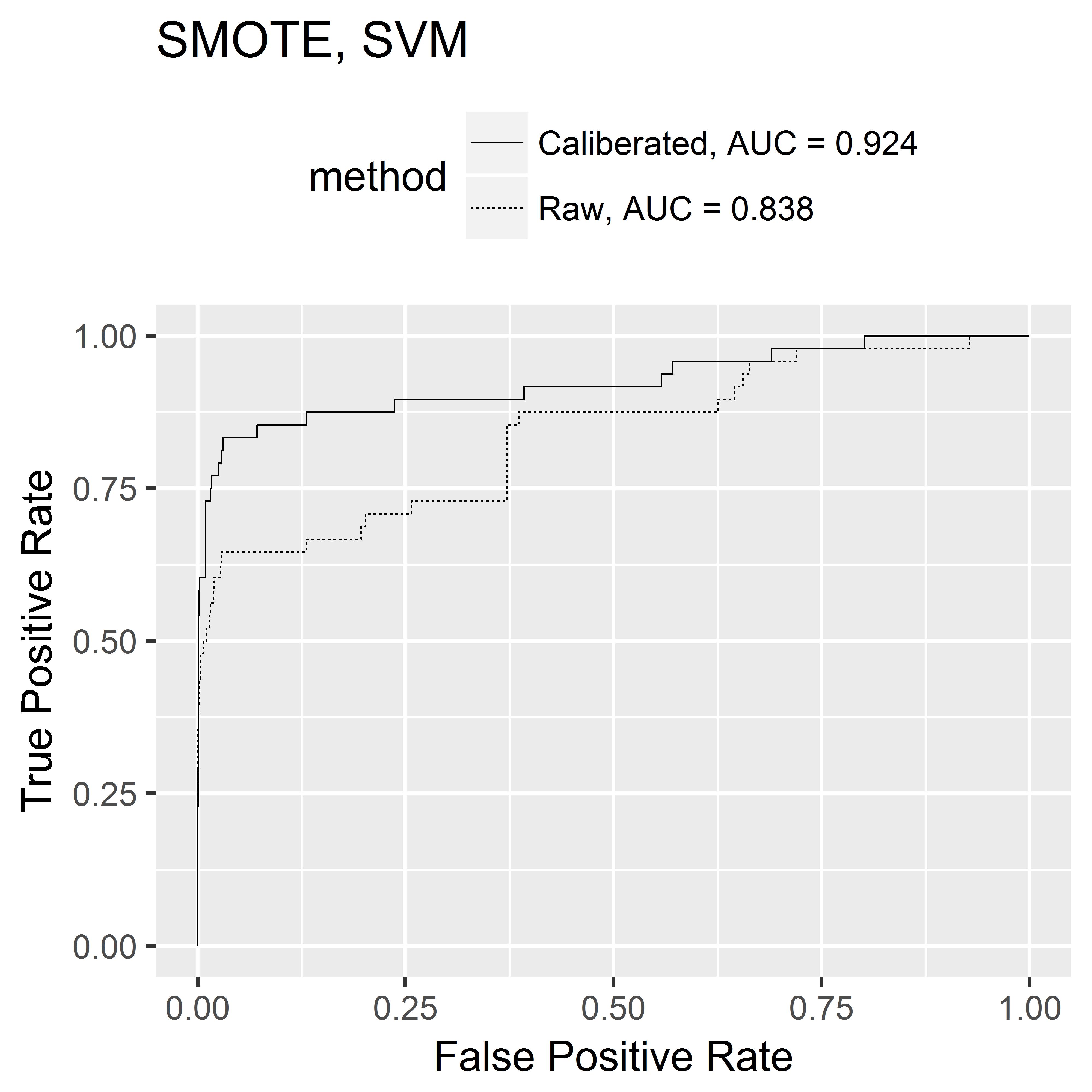} &
\includegraphics[width=0.33\textwidth]{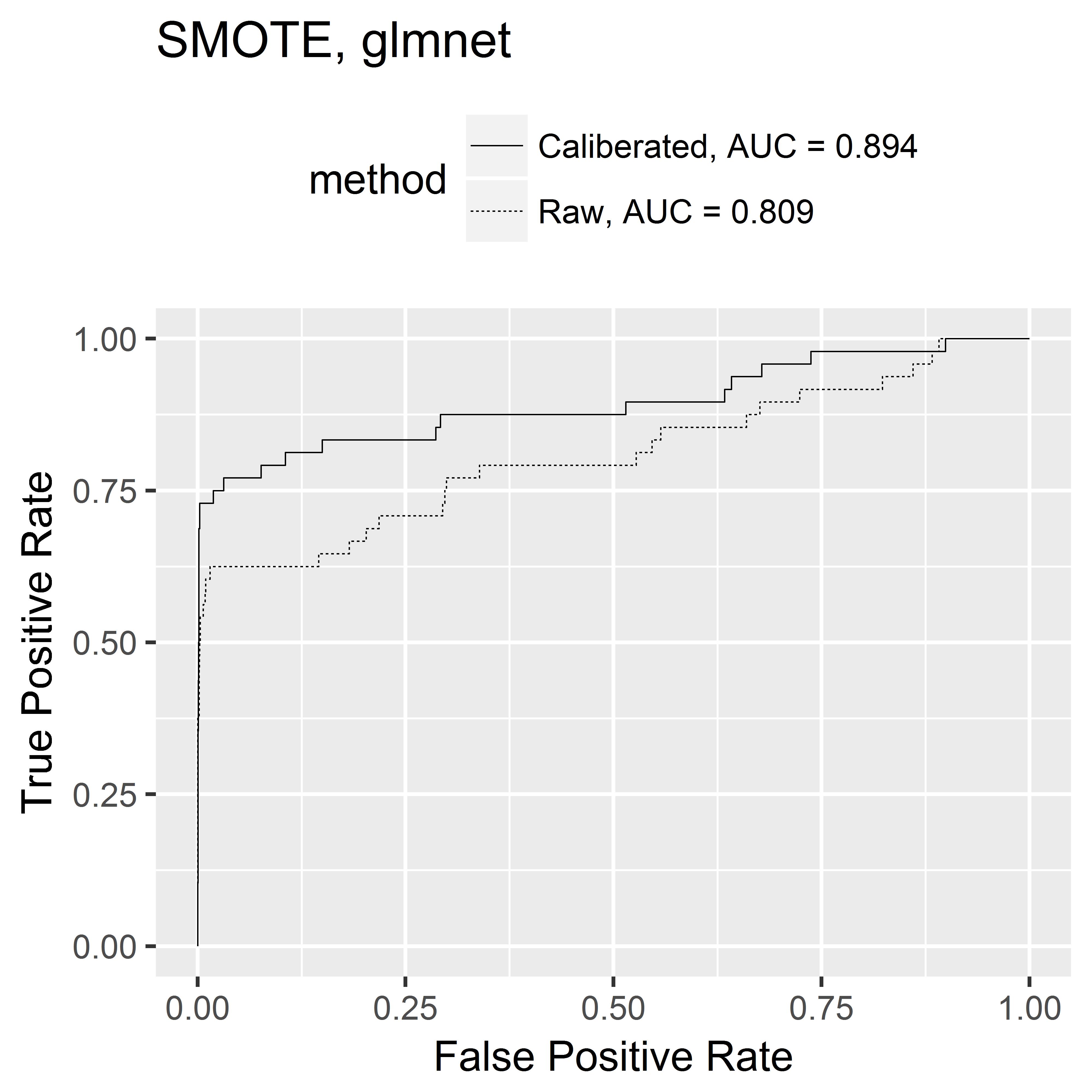} \\
\end{tabular}
\caption{Comparing ROC curves between calibrated scores and raw scores for the testing data.}\label{fig:stratasampling}
\end{figure}

In this experiment, I constructed a stratified sample using both training and tuning data set, used such a sample for model training, and evaluated the performance on the testing data. Specifically, the full data set was stratified into two segments based on whether $V12>0$. Within each segment, an independent model was developed using some data unbalancing technique. I evaluated three data unbalancing methods: SMOTE, observation weighting and under-sampling, as well as three popular classification models implemented in R: gradient boosting (xgboost, \cite{chen2016xgboost}), SVM (\cite{dimitriadou2011e1071}) and logistic regression with elastic net penalty (glmnet, \cite{friedman2009glmnet}). For each model, I first used the tuning data set to select the tuning parameters via random search. Using the selected tuning parameters, I then re-fitted the model based on the training data and evaluate the ROC curve on the tuning data. Next, the tuning data ROC curve was smoothed by local polynomial regression (\cite{cleveland1992local}) with smoothness parameter 0.05, which gave the estimation for the slope function. The final model was fitted based on the combination of training and tuning set, and applied to the testing data set for scoring. These raw testing data classification scores, together with the estimated slope function, were applied to Eqn. \ref{eq:rx} to obtain the calibrated ranking values.

Figure \ref{fig:stratasampling} shows a typical outcome in my experiment, comparing the ROC curve between ranking by my proposed calibrated score and the raw classification scores on the testing data set. Across all classification models, and across all unbalancing techniques, the calibrated score ranking shows significant improvement over uncalibrated results. The average AUC of 50 repetitions of this experiment is shown in Table \ref{tbl:stratasample}, demonstrating that the performance improvement is consistent across all classification models and all data unbalancing techniques. This experiment demonstrates that when the classification model is trained based on stratified sample, calibration raw model scores is necessary, and the proposed method in this paper effective remedies the stratified sampling bias issue.

\subsection{Calibration for Unobserved Distribution Shifts}

\begin{figure}
\centering
\begin{tabular}{ccc}
\includegraphics[width=0.33\textwidth]{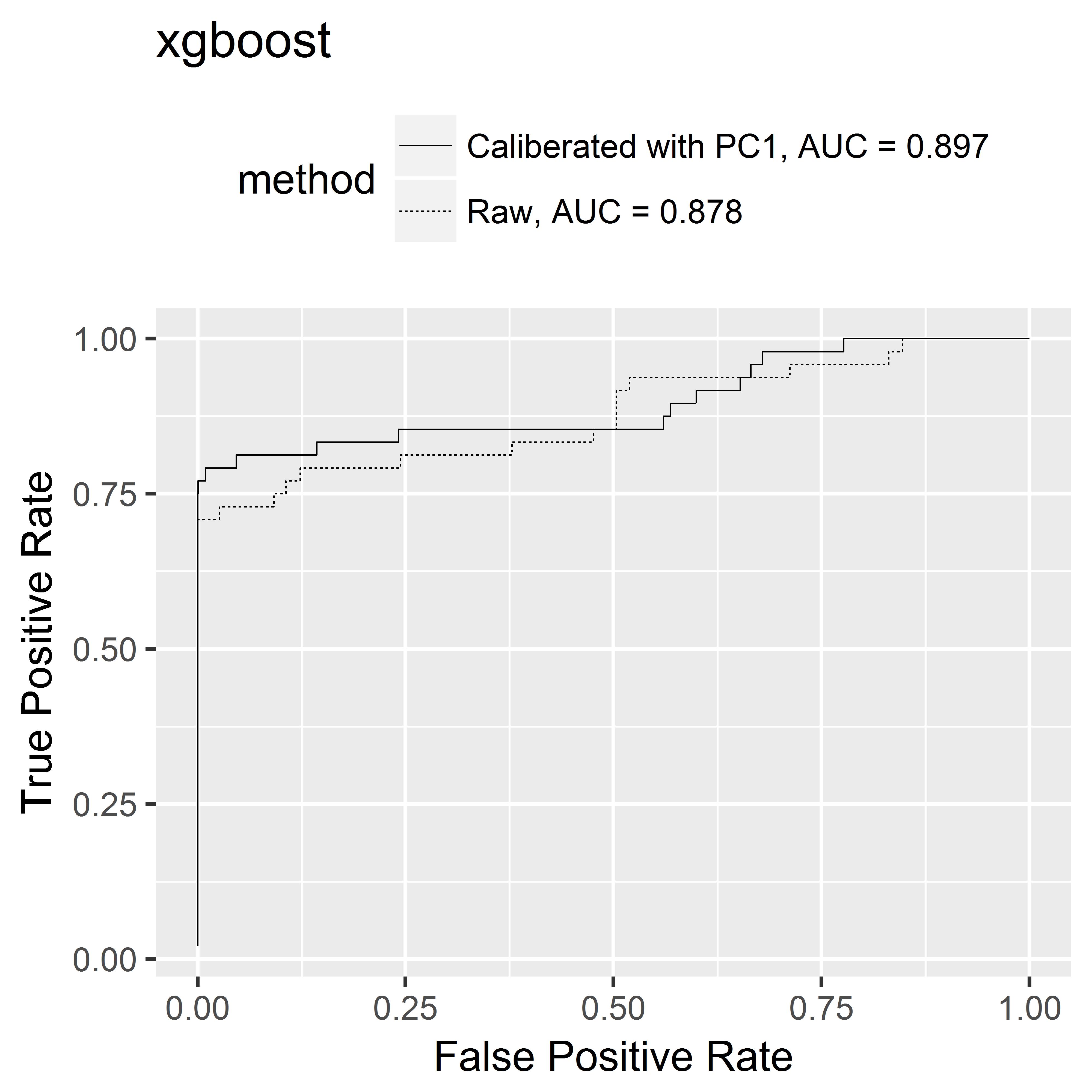} & 
\includegraphics[width=0.33\textwidth]{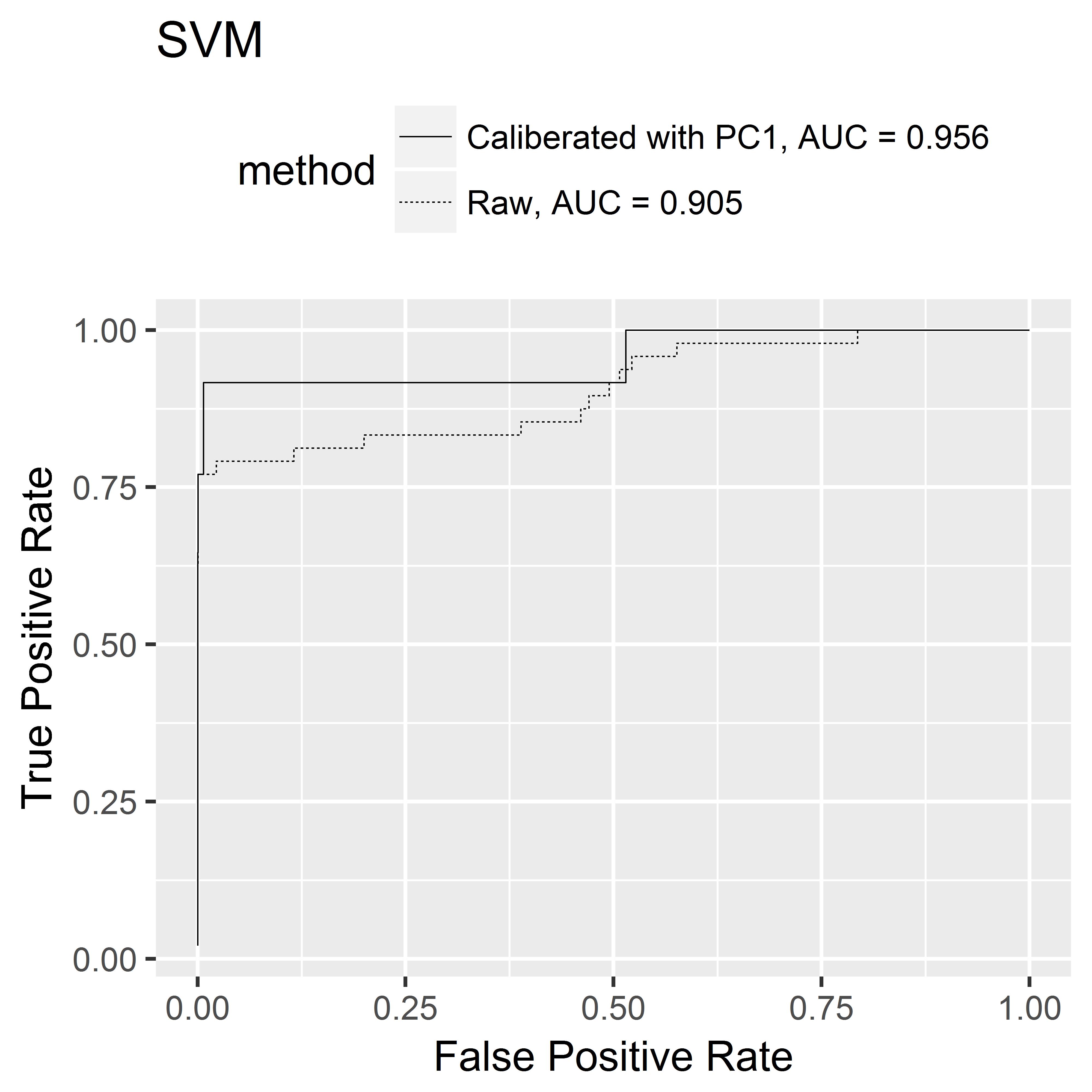} &
\includegraphics[width=0.33\textwidth]{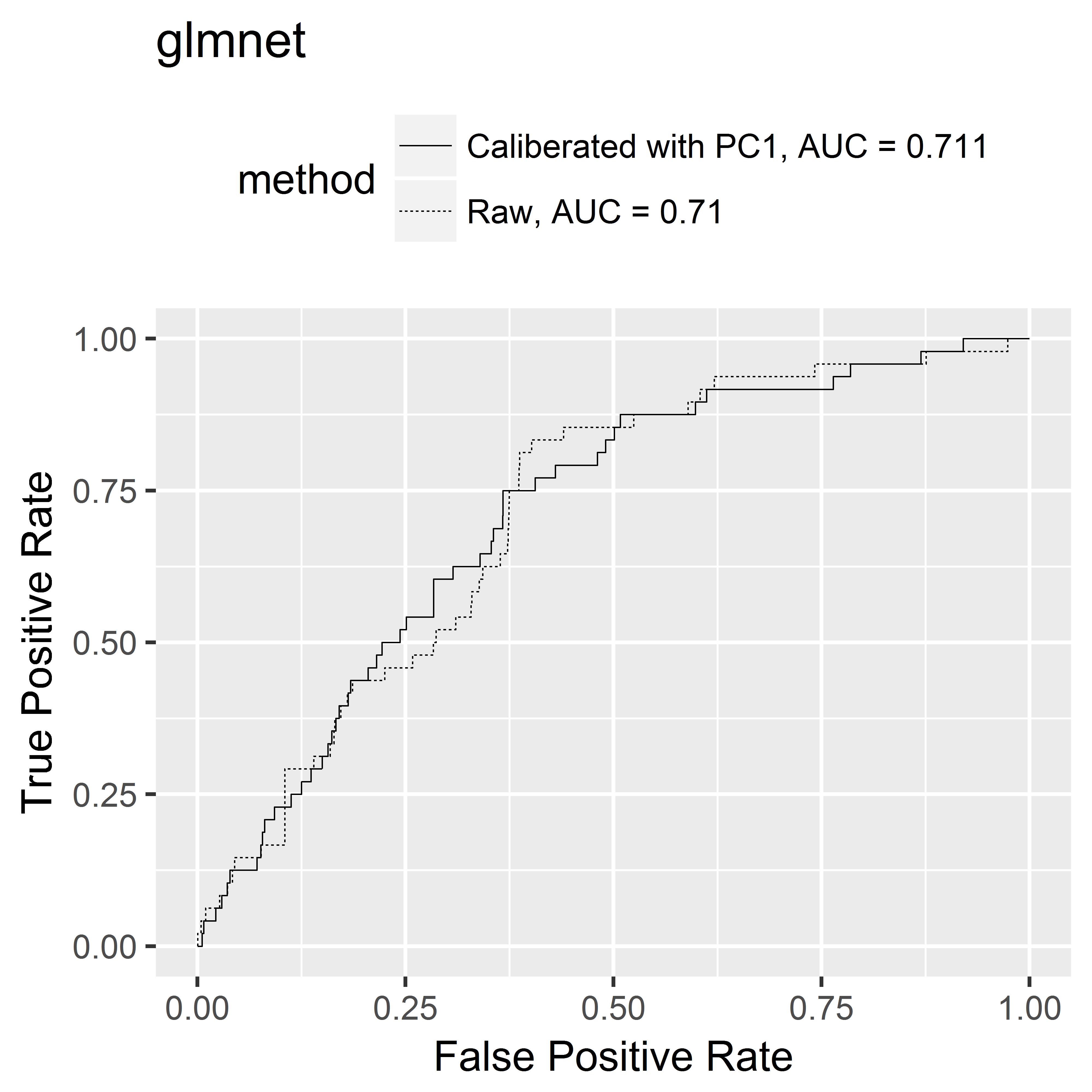} \\
\end{tabular}
\caption{ROC curves to compare calibrated models based on stratified data using PC1 and models using the raw data.}\label{fig:exp2}
\end{figure}

The second experiment was designed to evaluate whether constructing strata variables could improve classification performance, as discussed in Section \ref{sec:semi}. In this experiment, I first fitted a PCA model for all predictors across the full data, and segmented the full data into two groups based on the sign of the first PC. Then, a classification model was fitted on each of the two groups, without data unbalancing techniques. As in the first experiment, tuning data were used for both tuning parameter selection and estimating the slope function. The ranking scores for the two groups in the testing data set were calculated independently before being combined for overall ranking. As the benchmark, I also fitted a classification model using all the training and tuning data without segmentation. In this experiment, the above process is repeated for three classification models: xgboost, SVM and glmnet. Figure \ref{fig:exp2} compares the performance between models with PC1 calibration and the benchmark. Across all the three models, using PC1 to stratify the data for model building improved the ROC compared to the benchmark, and the improvement was the most significant for SVM models. For proof of concept, I do not attempt to further optimize certain details in this approach, including finding the best strata variable and determining the optimal number of stratas. Nevertheless, the results in this experiment show that stratified model building with proper calibration has the potential to improve many classification models in practice.

\section{Conclusions}\label{sec:conclusions}

In this paper, I investigate the problem of calibrating classification scores with sampling bias, and give the optimal calibration solution for the specific case when sampling bias is stratified. In certain business settings, the challenge in building classification models is the sheer volume of sample size, and building models based on stratified samples is not uncommon practice. Lacking theoretical guidance, the importance of cross-strata calibration is not understood, and in practice it is often neglected or performed sub-optimally. This paper, to my knowledge, is the first one that investigates this problem and proposes a theoretically justified solution.

Through developing theoretical results in this paper, I also find several interesting relationships between topics that have rarely been linked together in the machine learning literature. First, discriminant models and generative models are usually considered as different domains. ROC curve is generally applied for discriminant models, but my results show that optimizing ROC ideally leads to solving for the conditional joint distribution of $\hat f(X),Y$. This is an interesting example where the two domains cross with each other. Second, for real data where sampling strata are not defined a-priori, my results suggest that unsupervised learning might help to detect potential sampling shifts and mitigate their impact on classification. This generates an idea that is connected to semi-supervised modeling. Third, by considering extreme stratified settings the proposed method also implies new ways of ensembling models. Such relationships have not been explored in literature, yet they may yield insights into these related research areas that may warrant future research.

\bibliographystyle{plain}
\bibliography{classification_calibration}

\appendix

\end{document}